\documentclass[runningheads,a4paper]{llncs}

\ifx\undefined\technicalreport
  \newcommand\tr[1]{}
  \newcommand\ntr[1]{#1}
\else
  \newcommand\tr[1]{#1}
  \newcommand\ntr[1]{}
\fi

\usepackage{amssymb}
\usepackage{graphicx}

\usepackage{amsmath}
\usepackage{amssymb}
\usepackage{graphicx}
\usepackage{theorem}
\usepackage{float}
\usepackage{algorithm}
\usepackage[noend]{algpseudocode} 
\usepackage{subfigure}
\usepackage{color}
\usepackage{colortbl}
\usepackage{xspace}

\renewenvironment{proof}{\noindent {\bf Proof: }}{\hfill\proofend\linebreak}

\newcommand{\proofend}{$\Box$}

\newcommand{\e}{\text{e}}

\newcommand{\RNum}[1]{\uppercase\expandafter{\romannumeral #1\relax}}

\usepackage{color}

\usepackage{url}
\urldef{\mailsa}\path|{alfred.hofmann, ursula.barth, ingrid.haas, frank.holzwarth,|
\urldef{\mailsb}\path|anna.kramer, leonie.kunz, christine.reiss, nicole.sator,|
\urldef{\mailsc}\path|{janson, schindel}@informatik.uni-freiburg.de|    
\newcommand{\keywords}[1]{\par\addvspace\baselineskip
\noindent\keywordname\enspace\ignorespaces#1}

\begin{document}

\mainmatter  

\newcommand\titleplain{Receiving Pseudorandom PSK}

\title{\titleplain}

\titlerunning{\titleplain}

%
%
\author{Thomas Janson \and Christian Schindelhauer}
\authorrunning{Thomas Janson \and Christian Schindelhauer}

\institute{University of Freiburg,\\
Germany\\ 
\mailsc\\
}

\toctitle{\titleplain}
\tocauthor{Thomas Janson, Christian Schindelhauer}
\maketitle

\begin{abstract}
Pseudorandom PSK \cite{jes15_mimo_in_time} enables parallel communication on the same carrier frequency and at the same time. We propose different signal processing methods to receive data modulated with pseudorandom PSK. This includes correlation with the carrier frequency which can be applied to signals in the kHz to MHz range and signal processing in the intermediate frequency where the correlation with the carrier frequency is performed analogous in the RF front end. We analyze the computation complexity for signal processing with the parameters of symbol length $T$ and number of repetitions of each symbol $K$ with pseudorandom PSK and show that the number of operations for each sampling point is $\Theta\left(K\right)$.

\end{abstract}


\section{Introduction} \label{s:introduction}
The proposed scheme for radio transmission of \cite{jes15_mimo_in_time} repeats each symbol $K$ times with pseudo random phase shifts at each repetition. If the receiver uses the same sequence of pseudo random phase shifts for correlation, it can demodulate the data from the signal. On the other hand, if a receiver uses another sequence of pseudo random phase shifts which is not applied by the sender, it receives noise with a smaller amplitude. This reduces the correlation of parallel and interfering transmissions on the same carrier frequency in expectation and makes possible to have parallel transmission at the same time, on the same carrier frequency, and possibly in the same spatial area. This is particularly interesting for ad hoc networks and distributed medium access control (MAC). 

According to \cite{jes15_mimo_in_time},  we can correlate the received signal $y$ for known synchronization with a carrier signal with frequency $f$ with
\begin{equation}
\rho\left(\tau\right) = \sum_{t=0}^{T-1}\text{e}^{-j2\pi f\cdot t/R}\cdot \sum_{k=0}^{K-1}y\left(k\cdot T+t_{k}\left(t\right)+t+\tau\right) \label{eq:corr_org}
\end{equation}
where  $\phi_{k}$ are pseudorandom phase shifts with $0\le k<K$, sample rate $R$, $T$ samples per symbol, and
\begin{equation}
t_{k}\left(t\right) = \begin{cases}
\frac{\phi_{k}\cdot R}{2\pi f} & \text{if }t\le T-\frac{\phi_{k}\cdot R}{2\pi f}\\
-T+\frac{\phi_{k}\cdot R}{2\pi f} & \text{otherwise}
\end{cases} \ . \label{eq:time_shift}
\end{equation}
In a real channel, we also encounter effects like attenuation due to path path and multi-path propagation which is not subject of the analysis in this paper.
Here, we show different signal processing methods how we can also efficiently synchronize to the transmitter signal besides demodulation of data after synchronization. This makes necessary correlate the received signal with the carrier (which is a sinusoidal carrier with pseudorandom phase shifts) at every sample position to determine the point of synchronization, i.e. start and end of a symbol. 

\subsection{Contribution}
We propose  correlation methods for the K-repetition scheme \cite{jes15_mimo_in_time} necessary for synchronization and demodulating data. We improve the number computations of a correlation at a sample time from $\Theta\left(K\cdot T\right)$ to $\Theta\left(T\right)$, where $K$ is the number of repetitions of each symbol and $T$ is the number of samples per symbol. Table~\ref{ta:corr_carrier_freq} shows the computational complexity of the proposed methods.
\begin{table}[htp]
\begin{center}\begin{tabular}{c|c|c|c|c}
sampling freq. & method & additions & multiplications & phase error \\\hline
carrier & $\rho$ & $\Theta\left(K\cdot T\right)$ & $\Theta\left(T\right)$ & $0$ \\
carrier & $\rho_{\text{ma2}}$ & $\Theta\left(K\right)$ & $\Theta\left(K\right)$ & $0$ \\
carrier& $\rho_{\text{ma3}}$ & $\Theta\left(K\right)$ & $\Theta\left(1\right)$ & $<\frac{1}{p-1}$ \\
intermediate & $\rho_{\text{ma5}}$ & $3K-1$ & $K$ & $0$
\end{tabular} \caption{Correlation methods with carrier frequency $f$ with $T$ samples respectively $p$ periods per symbol and $K$ repetitions of each symbol with pseudo-random phase shifts}
\end{center}
\label{ta:corr_carrier_freq}
\end{table}
We present methods for signal processing at the carrier frequency or at a smaller intermediate frequency, where correlation with the carrier frequency is done analogously in the radio fronted.
For sampling at the carrier frequency we assume oversampling with $T$ samples per symbol and consider the correlation at each sample point.
The phase error relates only to computational errors of the signal processing method and not e.g. errors due to sampling.

\section{Correlation with Carrier Signal}
We show algorithms for signal processing with correlation of the received signal $y\left(\tau\right)$ at sampling point $\tau$ and a reference signal containing a sinusoidal carrier with frequency $f$. We suppose oversampling with a rate $R$ with $T$ samples per symbol.

In the first algorithm $\rho_{\text{ma1}}\left(\tau\right)$ for correlation at sample point $\tau$, we set $K=1$ and get PSK without repetition of the symbols. We assume that at time  $\tau=0$, the phase angle of the reference signal for correlation is zero with $\e^{-j2\pi f\cdot 0}=0$. We intend to compute the correlation of a window of $T$ samples for each sample point, i.e. for the windows with sample ranges $\left[0,T-1\right],\left[1,T\right], \dots, \left[\tau,\tau+T-1\right]$. If we compare the correlation sum at two subsequent points $\tau-1$ and $\tau$ with windows $\left[\tau-T,\tau-1\right]$ and $\left[\tau-T+1,\tau\right]$, the calculation only differs by the sample points $\tau-T$ and $\tau$ and the remaining $T-1$ summands are the same. Comparable to a moving average (ma), we can compute the correlation for the received signal $y\left(\tau\right)$ at sample point $\tau$ from the correlation result $\rho_{\text{ma1}}\left(\tau-1\right)$ at the preceding sample point with
\begin{equation}
\rho_{\text{ma1}}\left(\tau\right) = \rho_{\text{ma1}}\left(\tau-1\right) + y\left(\tau\right)\cdot \e^{-j2\pi f\cdot \tau/R} - y\left(\tau-T\right)\cdot \e^{-j2\pi f\cdot \left(\tau-T\right)/R} \ . \label{eq:corr_ma1}
\end{equation}
This equation is recursive and we can initialize $\rho_{\text{ma1}}\left(0\right) := 0$ and $y\left(\tau\right)=0$ for $\tau<0$.
If we buffer the intermediate results of $x\left(\tau\right)\cdot \e^{-j2\pi f\cdot \tau/R}$ in a ring buffer with $T$ elements (which is initialized with zeroes), we can save the computational costs for the third term  $x\left(\tau\right)\cdot \e^{-j2\pi f\cdot \left(\tau-T\right)/R}$. Then, we only need one multiplication and two additions of complex values for each sample point. Since we correlate the received signal $y\left(\tau\right)$ with a periodically carrier $\e^{j2\pi f \tau}$, the complex values can be pre-calculated and kept in a lookup table for given discrete sampling points.

\begin{lemma} \label{le:rho_ma1}
After $T$  sequential runs of $\rho_{\text{ma1}}\left(\tau\right) $ at $\tau = 0,\dots,T-1$, the operation $\rho_{\text{ma1}}$ outputs the same result as $\rho\left(\tau\right)$ of Equation~(\ref{eq:corr_org}) for $K=1$. 
\end{lemma}
\begin{proof}
For $K=1$, Equation~(\ref{eq:corr_org}) simplifies to
\begin{equation}
\rho\left(\tau\right) = \sum_{t=0}^{T-1} y\left(\tau+t-T+1\right) \cdot \text{e}^{-j2\pi f\cdot t/R} \ . \label{eq:rho_K1}
\end{equation} 
The original equation (\ref{eq:corr_org}) returns at $t=0$ which requires the reception of the samples at $\tau=\left[0,T\right)$ for simplification.
Without loss of generality we have added here the term $-\left(T-1\right)$ and get the first correlation output at $\tau=T-1$.
For the first $T$ runs, the term $y\left(\tau-T\right)$ of Equation~(\ref{eq:corr_ma1}) is zero. 
With $\rho_{\text{ma1}}\left(0\right) := 0$, we get after $T$ iterations a sum of $y\left(\tau\right)\cdot \e^{-j2\pi f\cdot \tau/R}$ with $\tau=\left[0,T\right)$ which is 
\begin{equation*}
\rho_{\text{ma1}}\left(T-1\right) = \sum_{t=0}^{T-1} y\left(t\right)\cdot \e^{-j2\pi f\cdot t/R} = \rho\left(T-1\right)\ .
\end{equation*}
The proof for $T\ge0$ follows by induction.
\begin{eqnarray*}
\rho_{\text{ma1}}\left(\tau+1\right) &=& \rho_{\text{ma1}}\left(\tau\right) \\
&&+ y\left(\tau+1\right)\cdot \e^{-j2\pi f\cdot \left(\tau+1\right)/R} - y\left(\left(\tau+1\right)-T\right)\cdot \e^{-j2\pi f\cdot \left(\left(\tau+1\right)-T\right)/R} \\
&=& \left(\sum_{t=0}^{T-1} y\left(\tau+t-T+1\right) \cdot \text{e}^{-j2\pi f\cdot t/R}\right) \\
&&+ y\left(\tau+1\right)\cdot \e^{-j2\pi f\cdot \left(\tau+1\right)/R} - y\left(\left(\tau+1\right)-T\right)\cdot \e^{-j2\pi f\cdot \left(\left(\tau+1\right)-T\right)/R} \\
&=& \sum_{t=1}^{T} y\left(\tau+t-T+1\right) \cdot \text{e}^{-j2\pi f\cdot t/R} \\
&=& \rho_{\text{ma1}}\left(\tau+1\right)
\end{eqnarray*}
\end{proof}

For a repetition of symbols with $K>1$, we can use a moving average as well with
\begin{equation}
\rho_{\text{ma2}}\left(\tau\right) = \sum_{k=0}^{K-1} \rho_{\text{ma1}}\left(\tau - T\cdot k\right) \cdot \e^{j\phi_k} \ . \label{eq:ma2}
\end{equation}
For each summand in Equation~(\ref{eq:ma2}) we need two additions and two multiplications of complex values. The overall computational complexity of $\phi_{\text{am2}}\left(\tau\right)$ is $\Theta\left(K\right)$.

\begin{lemma}
After $T\cdot K$  sequential runs of $\rho_{\text{ma2}}\left(\tau\right) $ at $\tau = 0,\dots,T-1$, the operation $\rho_{\text{ma2}}$ outputs the same result as $\rho\left(\tau\right)$ of Equation~(\ref{eq:corr_org}) for $K\ge1$. 
\end{lemma}
\begin{proof}
We have already proven the statement for $K=1$ in the proof of Lemma~\ref{le:rho_ma1}. To proof the case $K>1$ we insert Equation~(\ref{eq:rho_K1}) into (\ref{eq:ma2}) and get
\begin{eqnarray*}
\rho_{\text{ma2}}\left(\tau\right) 
&=& \sum_{k=0}^{K-1} \left(\sum_{t=0}^{T-1} y\left(\tau-T\cdot k+t-T+1\right) \cdot \text{e}^{-j2\pi f\cdot t/R} \right) \cdot \e^{j\phi_k} \ .
\end{eqnarray*}
We state in \cite{jes15_mimo_in_time} that the delay $t_k\left(t\right)$ replaces the phase shift with operation $ \e^{j\phi_k}$ and hence $\rho_{\text{ma2}}\left(\tau\right)  = \rho\left(\tau\right)$. Please note, that we assumed for the computation of $\rho\left(\tau\right)$ to have the samples for $\left[\tau,\tau+K\cdot T\right)$ in advance and in $\rho_{\text{ma}2}$ we compute the result from the preceding samples $\left(\tau-K\cdot T, \tau\right]$ which taken into account in the last transformation of the equation. Consequently, we need $T\cdot K$ iterations before  $\rho_{\text{ma2}}\left(\tau\right) $ is initialized and outputs the correlation of a window of length~$T\cdot K$.
\end{proof}

In the next calculation procedure $\phi_{\text{ma3}}\left(\tau\right)$, we reduce the number of multiplications from $2\cdot K$ to one with the sacrifice of a precision loss.
\begin{eqnarray}
\rho_{\text{ma3}}\left(\tau\right) &=& 
\rho_{\text{ma3}}\left(\tau-1\right) \nonumber\\
&&+ \left(\sum_{k=0}^{K-1} y\left(\tau + t'_k \right)\right) \cdot \e^{-j2\pi f\cdot \tau/R} \nonumber\\
&&- \left(\sum_{k=0}^{K-1} y\left(\tau -T + t'_k \right)\right)\cdot \e^{-j2\pi f\cdot \left(\tau-T\right)/R} \ . \label{eq:corr_ma3}
\end{eqnarray}
with pre-calculated values
\begin{equation}
t'_k := \frac{\phi_{k}\cdot R}{2\pi f}  \ . \label{eq:time_shift2}
\end{equation}
We assume that we compute the correlation for each sampling point $\tau=0,1,\dots$ sequentially. The third summand of Equation~(\ref{eq:corr_ma3})  has already been computed for the correlation at $\left(T-1\right)$ sample points before and we can buffer these intermediate results in a  ring buffer with $T$ elements to reduce the number of computational steps. Then we only have to compute the sum of the second summand at sample points $\tau$ which are $K$ additions, perform only one multiplication with $\e^{-j2\pi f\cdot \tau/R}$ for correlation (instead of $2\cdot K$ multiplications before), and compute two additional sums of complex values for the overall sum of the three summands of Equation~(\ref{eq:corr_ma3}).

We include the phase shifts $\phi_k$ of the repeated symbols with the time shift $t'_k$ in Equation~(\ref{eq:time_shift2}). The difference to Equation~(\ref{eq:time_shift}) is that we do not perform the modulo-operation and we will make a computational error at the end of each symbol where the the sample point $\tau + t'_k $ is in the following sample (instead of at the beginning of the correct symbol with the modulo operation). Let us estimate this error and assume that each symbol is send for $p$ periods of the carrier frequency, i.e. for time $p/f$ for carrier frequency $f$. In worst case, $\phi_1=0$ and $\phi_{k'}=2\pi$ for $2\le k'\le K$. Then we read almost all information of the last period from the subsequent next symbol. We can represent the correlation of $p$ periods as superposition of the correlation of $p$ individual correlations, where we know that the last is erroneous. 
\begin{figure}[hbt]
	\begin{center}
	\includegraphics[scale=0.6]{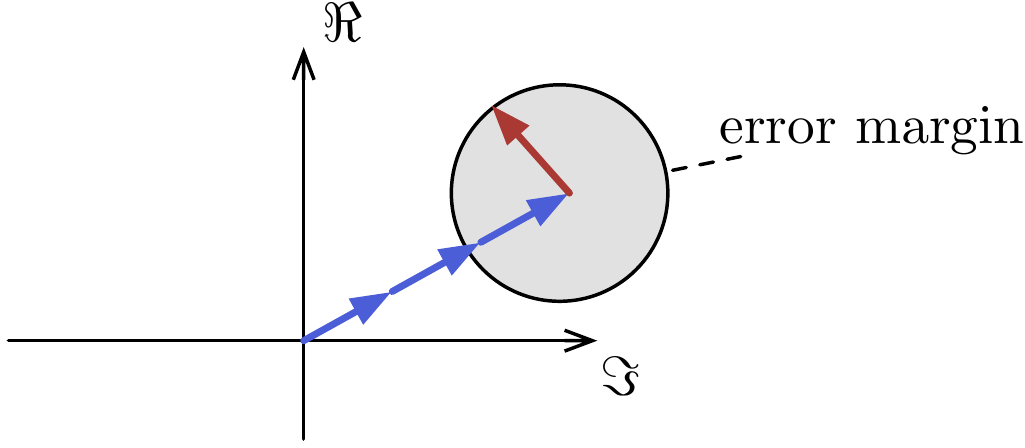}
	\end{center}
	\caption{Error margin of correlation method $\rho_{ma3}$ in Equation~(\ref{eq:corr_ma3}) for $P=4$ periods}
	\label{fig:corr_ma3_error}
\end{figure}
Figure~\ref{fig:corr_ma3_error} shows an example of a correlation with $p=4$ periods of the carrier where each vector illustrates the correlation of one period and the last correlation with the red-colored vector is erroneous. 
The error in the amplitude is at most $1/(p-2)$ if the erroneous correlation of the last period has a phase shift of $\pi$. The phase error is at most $\sin\left(\frac{1}{p-1}\right)\le \frac{1}{p-1}$ for $p\ge2$. With the sacrifice of the computation error we can first compute the sum of $K$ samples and then perform the multiplication with the complex value to compute the correlation. Compared to Equation~(\ref{eq:ma2}), this reduces the number of multiplications of a correlation from $2\cdot K$ to $1$.

\section{Processing at an Intermediate Frequency}

For carrier frequencies in the GHz range, oversampling and correlation with a rate larger than the carrier frequency is mostly too fast for current hardware and correlation is implemented in hardware. The output of the radio front end are complex values in an intermediate frequency which is slower than the carrier frequency. Let us assume that we have $G$ complex-valued samples $y_c\left(\tau\right)$ in the intermediate frequency. The discrete sample points $\tau$ are in this case in the intermediate frequency.

Just like in Equation~(\ref{eq:corr_ma1}) for a correlation with the carrier, we can use  once again a moving average at the intermediate frequency for $K=1$ with
\begin{equation}
\rho_{\text{ma4}}\left(\tau\right) = \rho_{\text{ma4}}\left(\tau-1\right) + y_c\left(\tau\right) - y_c\left(\tau-G\right) \ . \label{eq:corr_ma4}
\end{equation}
Since the correlation is performed  in analog hardware in the radio front end, this only needs two (complex) additions.

For $K>1$ we can use $K$ moving averages with
\begin{equation}
\rho_{\text{ma5}}\left(\tau\right) = \sum_{k=0}^{K-1} \rho_{\text{ma4}}\left(\tau - k\cdot G\right) \cdot \e^{j\phi_k} \ . \label{eq:corr_ma5}
\end{equation}
A correlation with $\rho_{\text{ma5}}\left(\tau\right)$ needs $3\cdot K-1$ complex additions and $K$ complex multiplications for the phase shifts with $\phi_k$.

\begin{figure}[hbt]
	\begin{center}
	\includegraphics[width=\textwidth]{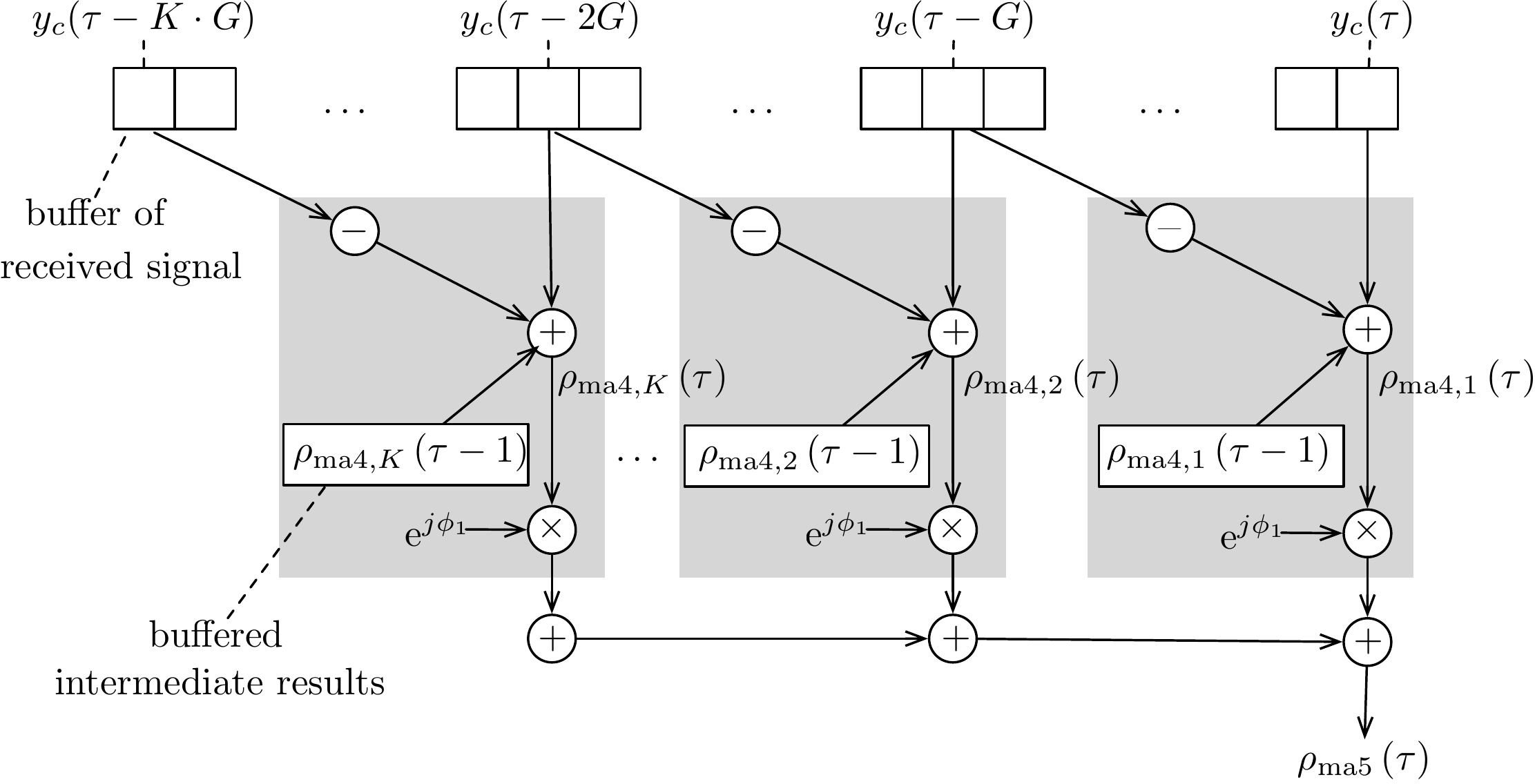}
	\end{center}
	\caption{Block diagram for signal processing of method $\rho_{\text{ma5}}\left(\tau\right)$ at sample point $\tau$ for $K$ repetitions of each symbol with pseudorandom phase shifts $\phi_k$ with $0\le k<K$}
	\label{fig:block_diagram_rho_ma5}
\end{figure}
The block diagram in Figure~\ref{fig:block_diagram_rho_ma5} shows the computation of $\rho_{\text{ma5}}\left(\tau\right)$. The computations in the grey boxes are independent from each other and can be computed in parallel. The total sum (at the bottom of the chart) depends on the intermediate results (in the grey boxes) and it is also possible to add these $K$ values with a tree-like concatenation with only $\log \left(K\right)$ instead of $K$ computational steps shown in the chart. Then the overall number of computational steps is only $\Theta\left(\log K\right)$ if we can perform $K$ operations in parallel.

\section{Conclusions}
We show that computational costs for  synchronizing and demodulating data in the K-repetition scheme are only increased by factor $\Theta\left(K\right)$ compared to signal processing of a modulation without repeating each symbol $K$ times.
This indicates the practical applicability of the K-repetition scheme.  We show  methods for correlation with the carrier (a sinusoidal signal with carrier frequency $f$ and $K$ pseudorandom phase shifts). These are mostly applicable for carrier frequencies in the kHz to MHz range where the clock speed of a processor is high enough. We also show methods for demodulating the K-repetition scheme at an intermediate frequency where correlation with the carrier is done at the radio front. This is the typical application for carrier frequencies in the GHz range and needs $\Theta\left(K\right)$ operations per sampling point. If we can run $K$ operations in parallel, we also show that the number of computational steps is only $\Theta\left(\log K\right)$.

\bibliographystyle{abbrv}
\bibliography{mimo}

\end{document}